\documentclass[journal,,draftclsnofoot,letterpaper,onecolumn]{IEEEtran} %,
\usepackage{amsmath}
\usepackage{amssymb}
\usepackage{amsthm}
\newtheorem{pavikc}{\textbf{Corollary}}

\newtheorem{pavikp}{\textbf{Proposition}}
\usepackage{bm}
\usepackage{color}
\usepackage{ulem}
\usepackage[pdftex]{graphicx}
\usepackage{tikz}
\usepackage[numbers,square,sort&compress]{natbib}
\usetikzlibrary{arrows,shapes}

\title{Secure Transmission in Amplify and Forward Networks for Multiple Degraded Eavesdroppers}
\author{\IEEEauthorblockN{Siddhartha Sarma, %\IEEEauthorrefmark{1},
Samar Agnihotri
and
Joy Kuri
}
\thanks{S. Sarma and J. Kuri are with the Department of Electronic Systems Engineering, Indian Institute of Science, Bangalore, Karnataka - 560012, India (e-mail: \{siddharth, kuri\}@dese.iisc.ernet.in).}
\thanks{S. Agnihotri is with the School of Computing and Electrical Engineering, Indian Institute of Technology Mandi, Mandi, Himachal Pradesh - 175001, India (e-mail: samar@iitmandi.ac.in).}
}
\begin{document}

% make the title area
\maketitle

\begin{abstract}
%\boldmath
We have evaluated the optimal secrecy rate for Amplify-and-Forward (AF) relay networks with multiple eavesdroppers. Assuming i.i.d. Gaussian noise at the destination and the eavesdroppers, we have devised technique to calculate optimal scaling factor for relay nodes to obtain optimal secrecy rate under both sum power constraint and individual power constraint. Initially, we have considered special channel conditions for both destination and eavesdroppers, which led us to analytical solution of the problem. Contrarily, the general scenario being a non-convex optimization problem, not only lacks an analytical solution, but also is hard to solve. Therefore, we have proposed an efficiently solvable \textit{quadratic program} (QP) which provides a sub-optimal solution to the original problem. Then, we have devised an iterative scheme for calculating optimal scaling factor efficiently for both the sum power and individual power constraint scenario. Necessary figures are provided in result section to affirm the validity of our proposed solution.    
\end{abstract}
\section{Introduction}
Recently a significant amount of research is going on to ensure secure communication in wireless networks. Due to broadcast nature of wireless transmission, the transmitted messages can be intercepted by eavesdroppers. Though cryptographic security can be used to counteract eavesdropping, but secrecy measure of such scheme relies on the computational complexity of cryptographic functions rather than information theoretic principles. Also, distributing secret key across the network has its own overhead. On contrary, physical layer security schemes exploit the inherent randomness present in the wireless channel and provide information theoretically provable secure communication irrespective of the computational capability of the eavesdropper(s).
\par Physical layer security came to existence when Wyner in a seminal paper \cite{wyner} showed that a non-zero secrecy rate is possible for a discrete memoryless channel if the eavesdropper's channel is degraded.
%Wyner showed his result for single eavesdropper.
% \todo{single or multiple eavesdroppers?}. 
Following his work researchers have evaluated secrecy capacity and equivocation region of Single antenna and Multi-antenna systems. However, resource constrained multi-hop networks have not got enough attention though they are practically significant.
\par In a multihop wireless network, intermediate relay nodes have to follow certain relaying strategy for forwarding packets to the next relay or destination. Amplify and forward relaying scheme is simplest among them where each node transmits the message it has received after amplification (scaling). Though simplest in nature but the significance of this scheme lies in its low cost implementation and effectiveness against fading. Nevertheless, from theoretical point of view study of such a relaying scheme can help us to estimate lower bounds of the channel capacity of other communication scenarios (e.g. Analog Network Coding). Very recently, researchers have started investigating the significance of amplify and forward relaying for attaining physical layer security \cite{dong,zhang10}.
\par In our paper we consider a scenario where relay nodes uses amplify and forward relaying to convey the source message to the destination. However, due to the presence of one or more eavesdropper secrecy of communication is in jeopardy. For such a scenario secrecy rate of the network provide a good measure of performance of the system. Unlike some previous works where only total relay power constraints is assumed, we consider the individual relay power constraint also. In practice the relay nodes are generally powered by their individual power source without any means to share their power sources (e.g. battery). Therefore, individual relay constraint is more relevant in practical situations and general. Assuming the availability of global channel state information (CSI), we consider a two hop network consists of a single source, a single destination and multiple relay nodes. As each of the relay node connects the source and destination separately, they form a diamond like structure and hence named accordingly. We begin our analysis with a symmetric diamond network and provide the analytical solution for optimal scaling factor of relay nodes. We then relax the symmetric network assumption and analyze the scenario where eavesdropper's channel vector is scaled version of receiver's channel. For general case where multiple eavesdroppers are involved we have multiple secrecy rate corresponding to each eavesdropper and the objective would be to maximize the minimum of them over the same constraint set. In our paper, we provide a sub-optimal solution for individual relay constraints, whereas we propose an iterative algorithm for secrecy rate in case of sum constraint and individual constraint. 
%Now, in multiple eavesdroppers scenario we have multiple secrecy rate corresponding to each eavesdropper and the objective would be to maximize the minimum of them over the same constraint set. Therefore, the complexity of finding the optimal scaling factor for secrecy rate increases with the number eavesdropper we provide two SNR based alternate approach. 
%\par In practical receivers the correct decoding of the transmitted signal is dependent on the received signal to noise ratio (SNR). If the received SNR is below certain threshold then for all practical purposes we can assume that the receiver has not received any valid information. This motivates our alternate and weaker formulations where we consider two different optimization problems. As the relay power is a scarce resource, so we minimize the total relay power with two constraints a) SNR at legitimate receiver is above predefined threshold and b) SNR at all the eavesdropper is below some threshold. In other formulation, we consider the maximization of the SNR at the receiver while keeping SNR at the eavesdroppers below threshold. For both the formulations we impose individual relay constraints.
%    For such a scenario we are looking for the optimal amplification vector for the relay nodes so that secrecy rate is maximized.
\par We summarize our contribution as follows:
\begin{itemize}
\item For symmetric diamond network we provide analytical solution for secrecy rate.
\item We discuss and analyze a step-by-step procedure for calculating optimal secrecy rate when eavesdropper's channel vector is scaled version of receiver's channel. 
\item For general case we discuss the sub-optimal ``zero-forcing" solution for individual relay constraint. There we reformulate the optimization problem as a quadratic program which can be solved efficiently.
\item We propose an iterative algorithm for obtaining optimal secrecy rate for sum relay and individual constraint scenarios. 
%\item We then formulate a power minimization problem with SNR constraints on receiver and eavesdroppers. A relaxed semidefinite program is devised to tackle the non-convexity of the original problem.
%\item We formulate a receiver SNR maximization problem with constraint on eavesdropper's SNR. We then perform the transformation of the variables to obtain a convex problem from the original non-convex problem. 
\end{itemize}
\subsection*{Organization}
Our paper is organized as follows: In Section~\ref{sec:rel}, we survey the related work. The system model and notations are introduced in Section~\ref{sec:model}.
In Section~\ref{sec:secrate}, we analyze the secrecy rate for several channel conditions of receiver and eavesdropper.
% We then provide the alternate formulation in Section~\ref{sec:alternate}. 
We illustrate numerical results of the formulations in Section \ref{sec:res}. We conclude our paper in Section \ref{sec:concl} with a brief summary and possible future work.
\section{Related Works}\label{sec:rel}
Amplify and forward scheme was introduced by Schein and Gallager \cite{schein2001phd} and was considered as a mean of cooperative communication by \cite{laneman}, \cite{zhao}, \cite{borade}. Later several researchers have reported that cooperative scheme like amplify and forward not only provides robustness against channel variations but also ensures non-zero secrecy rate in certain scenarios where otherwise it is zero. For example, if  the source to destination channel is poor as compared to source to eavesdropper channel, then by using appropriate scaling factor in relay nodes we can cancel out the received signal at eavesdropper and thereby improve the secrecy rate. As we have assumed gaussian channel model it is worth mentioning that the secrecy capacity for Gaussian Wiretap channel was evaluated by Cheong and Hellman \cite{Leung}. Later the effect of fading on gaussian wiretap channel model was analysed by \cite{liang2008} and \cite{gopala}. The wire-tap model in context of multi-antenna system was considered and analysed by \cite{parada}, \cite{khisti}, \cite{shaf} and \cite{li2009}. But both the single and multi-antenna system were limited to single-hop network. As the multi-hop wireless networks are equally significant, so recently research in this area has got a good pace.  Lai and Gamal \cite{lai07} evaluated the secrecy capacity of relay-eavesdropper channel for different cooperative schemes and also evaluated corresponding equivocation region. Authors in \cite{dong09} reported the improvement in physical layer security with the help of cooperating relays. Same authors later elaborated the significance of amplify and forward scheme for attaining physical layer security in \cite{dong}. But in their work they considered total relay constraint criteria and provided bounding results for secrecy capacity. For multiple eavesdroppers scenario they suggested the ``zero-forcing solution" where by beam-forming the transmitted signal is nullified at each eavesdropper. The more practical individual relay constraint criteria was considered in \cite{zhang10}, \cite{yang2013cooperative}. The authors of the both the papers provided an iterative algorithm for calculating the optimal amplification vector for relay nodes for maximizing secrecy rate using semi-definite relaxation. Our work discusses the special cases of the model considered in \cite{yang2013cooperative} and investigate the nature of the solution for those special cases. Our approach significantly differs from the techniques used in \cite{zhang10} and \cite{yang2013cooperative} as we have converted our problem into a convex optimization problem by using a noble transformation of variables. Further we discuss the convergence of the solution to the global optimum.
%Unlike the semi-definite relaxation based solution them we propose an iterative algorithm which does not rely on any relaxation and uses a much simpler optimization formulation.
 Infact we identify that for certain special cases we can evaluate the optimal scaling factor analytically. Those analytical results are motivated from \cite{jing} and \cite{agnihotri11}, where authors have devised schemes to find the optimal amplification vector for attaining the capacity of amplify and forward network under individual relay constraints in absence of secrecy criteria. 
%In the subsequent sections we discuss about SNR based optimization formulation. Similar kind of power minimization problem for decode and forward relaying was considered in \cite{sarma}. But unlike decode and forward, in amplify and forward the noise get added at every stage of reception and lead to further complications.
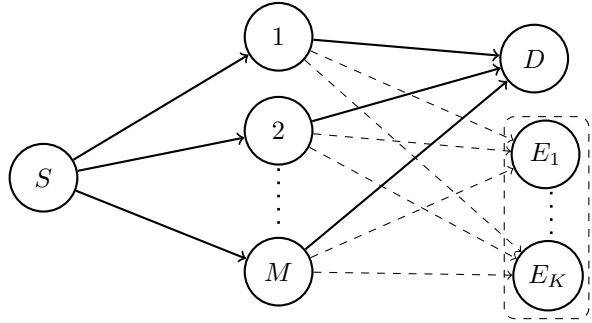
\begin{figure}[!h]
\centering
\begin{tikzpicture}[scale=1.25]
\tikzstyle{every node}=[draw,shape=circle,minimum size=0.9cm,style=thick]
\node (v0) at (180:2.5) {$S$};
\node (v1) at (90:1.5) {$1$};
\node (v2) at (90:0.5) {$2$};
\node (v3) at (270:1) {$M$};
\node (v4) at (25:3) {$D$};
\node (v5) at (5:2.85) {$E_1$};
\node (v6) at (340:3.05) {$E_K$};
\draw[style=thick,->] (v0) -- (v1);
\draw[style=thick,->] (v0) -- (v2);
\draw[style=thick,->] (v0) -- (v3);
\draw[style=thick,->] (v1) -- (v4);
\draw[style=thick,->] (v2) -- (v4);
\draw[style=thick,->] (v3) -- (v4);
\draw[dashed,->](v1) -- (v5);
\draw[dashed,->](v2) -- (v5);
\draw[dashed,->](v3) -- (v5);
\draw[dashed,->](v1) -- (v6);
\draw[dashed,->](v2) -- (v6);
\draw[dashed,->](v3) -- (v6);
\draw[loosely dotted,line width=1](0,0.10) -- (0,-.6);
\draw[loosely dotted,line width=1](2.9,-0.15) -- (2.9,-.6);
\draw[dashed,,rounded corners] (2.4,0.65) rectangle (3.3,-1.5);
%\draw[step=.5cm] (-3,-2) grid (2,2);
\end{tikzpicture}
\caption{Simple AF network with multiple ($K>1$) eavesdroppers}
\label{fig:model}
\end{figure}
%\begin{table}[p]
%\begin{tabular}{|cc|}
%\hline
%Symbols & Descriptions\\
%\hline
% $h_{.,.},\;g_{.,.}$&channel gain  \\ 
%$x,\;y$ & input and output signals  \\
%$z$ & Gaussian noise\\
%$\beta$ & amplification factor\\
%$P$ & Power values\\
%$\sigma^2$ & noise variances\\
%$\gamma$ & SNR\\
%$M,\;K$ & No. of relays and eavesdroppers\\
%$\nu,\;\mu$ & inverted channel gain\\
%$R_s$ & secrecy rate\\
%$\alpha$ & constant fraction as channel multiplier\\
%$\omega_i$ & amplified channel gain\\
%$\Psi,\;\varrho,\;r$ & function of $\omega$\\
%$p,\;q,\;s,\;u,\lambda$ & used in context of Asymmetric Network\\
%$\mathbf{w}$ & vector zero forcing solution\\
%$\mathbf{H}$ & Matrix function of channel gains\\
%$\eta$ & SNR threshold\\
%$\mathbf{D,C}$ & diagonal and Eavesdroppers' matrix \\
%$\kappa$ & total secrecy rate constraint\\
%$\tau$ & golden section parameter\\
%$\delta$ & tolerance parameter for iteration.\\
%\hline
%\end{tabular} 
%\caption{Table of Symbols for personal reference}
%\end{table}
\section{System Model}\label{sec:model}
The system model consists of a single source, a single destination, $M$ relay nodes and eavesdroppers as shown in figure \ref{fig:model}. The channel gain from source node to the $i^{th}$ relay node is denoted by a real constant $h_{s,i}$. Similarly channel gain from the $i^{th}$ relay node to destination or to eavesdropper is denoted by $h_{i,d}$ and $h_{i,e}$, respectively. Now, if we consider discrete time instants and neglect the transmission delays, then signal received at each relay node due to the source can be expressed as:
\begin{equation}\label{eq:rel}
y_i[n]= h_{s,i}x_s[n] + z_i[n] 
\end{equation} 
where $x_s[n]$ is the channel input at time instant $n$ and $z_i[n]$ is the noise at relay $i^{th}$ node. We assume that $\{z_i[n]\}, -\infty <n<\infty$ forms an  i.i.d. sequence of Gaussian random variables with zero mean and variance $\sigma^2$, \textit{i.e.} $z_i[n] \sim \mathcal{N}(0,\sigma^2)$ which is independent of the input signal at that receiver. In case of AF scheme each relay node scales its received signal before transmitting. The maximum scaling factor is determined by the individual power constraint of relay node and the received signal power at that relay node. We assume a power constraint over transmitted signal from each node which can be expressed as:
\begin{equation*}
E[x_i^2[n]] \le P_i,\quad  -\infty <n<\infty,\quad i \in \{s,1,2,\dots,M\}
\end{equation*}
So, transmitted signal from each relay node can be written as:
%\begin{equation}
%\begin{align}
\begin{align}\label{eq:beta}
 x_i[n+1]=\beta_iy_i[n],\quad -\beta_{i,max} \le\beta_i \le \beta_{i,max}
\mbox{ where } \beta_{i,max}^2= \frac{P_i}{h_{s,i}^2 P_s + \sigma^2}
\end{align}

%\end{align}
%\end{equation} 
\par Now, we can express the received signal at destination and eavesdroppers in following manner:
\begin{align}
y_d[n] = & \sum\limits_{i=1}^{M}h_{i,d}x_i[n]+z_d[n] \label{eq:dest}\\
y_k[n] = & \sum\limits_{i=1}^{M}h_{i,k}x_i[n]+z_k[n] \label{eq:eve}
\end{align}
Here $z_d[n]$ and $z_k[n]$ are mutually independent i.i.d. random variables distributed according to $\mathcal{N}(0,\sigma^2)$ and also independent of $z_i[n]$.
For the layered network shown in figure \ref{fig:model} both source signal ($x_s[.]$) and noise signals ($z_{(.)}[.]$) arrive at destination or eavesdropper traversing different but same respective delayed path. Therefore, the received signal at destination and eavesdroppers are free of intersymbol interference. So, we can omit the time indexing and use equation \eqref{eq:rel}, \eqref{eq:beta} and \eqref{eq:dest} to write the following expression.
\begin{equation}
y_d=\sum\limits_{i=1}^{M}h_{s,i}\beta_ih_{i,d}x_s + \sum\limits_{i=1}^{M}\beta_ih_{i,d}z_i + z_d
\end{equation}
In similar manner the received signal at eavesdropper can be written using equation \eqref{eq:rel}, \eqref{eq:beta} and \eqref{eq:eve}.
\begin{equation}
y_k=\sum\limits_{i=1}^{M}h_{s,i}\beta_ih_{i,k}x_s + \sum\limits_{i=1}^{M}\beta_ih_{i,k}z_i + z_k,\;k \in \{1,2,\cdots K\}
\end{equation}
The secrecy rate at destination for such a network model can be written as \cite{wyner}:  
\begin{equation}
R_s(P_s)=\underset{k}{\min}\; [I(x_s;y_d)-I(x_s;y_k)],\; k \in \{1,2,\cdots K\}
\end{equation}
where $I(x_s;y)$ represents the mutual information between random variable $x_s$ and $y$. \textit{Secrecy capacity} is defined as the maximum achievable secrecy rate over all the distribution of source symbol and scaling vector $\boldsymbol{\beta}$. Now, due to \cite{Leung} we know that secrecy capacity is attained for Gaussian channels when the inputs are distributed according to $\mathcal{N}(0,P_s)$ where $\mathbf{E}[x_s^2] = P_s$. Therefore, optimal secrecy rate can be written as following optimization problem.
\begin{subequations} \label{eq:opt1}
\begin{align}\label{eq:secrate}
& R_s^*(P_s)=  \underset{\boldsymbol{\beta}}{\max}\underset{k \in \{1,2,\dots,K\}}{\min}\;\left[R_d(P_s,\boldsymbol{\beta}) - R_k(P_s,\boldsymbol{\beta})\right]\\
& =  \underset{\boldsymbol{\beta}}{\max}\underset{k \in \{1,2,\dots,K\}}{\min}\;\left[ \frac{1}{2}\log\left(1+SNR_d \right)-\frac{1}{2}\log\left(1+SNR_k \right) \right]
\end{align}  
\end{subequations}
 where $SNR_l=  \frac{\left( \sum\limits_{i=1}^{M} h_{s,i}\beta_ih_{i,l}\right) ^2}{1 + \sum\limits_{i=1}^{M}(\beta_i h_{i,l})^2}\frac{P_s}{\sigma^2}$. 
 %and $SNR_e=\frac{\left( \sum\limits_{i=1}^{M} h_{s,i} \beta_i h_{i,k}\right) ^2}{1 + \sum\limits_{i=1}^{M}(\beta_ih_{i,k})^2}$
\par If we use the following vector and matrix notations then SNR at destination or eavesdropper can be represented as  
\begin{align}\label{eq:snr_vec}
SNR_l= \frac{\left( \sum\limits_{i=1}^{M} h_{s,i}\beta_ih_{i,l}\right) ^2}{1 + \sum\limits_{i=1}^{M}(\beta_i h_{i,l})^2}\frac{P_s}{\sigma^2}=\frac{(\mathbf{h_{s,l}}^t\bm{\beta})^2}{1+\bm{\beta}^tdiag(\mathbf{h_l})\bm{\beta}}\frac{P_s}{\sigma^2} 
\end{align}
where 
$\mathbf{h_{s,l}}=[h_{s,1}h_{1,l},h_{s,2}h_{2,l},\cdots,h_{s,M}h_{M,l}]^t$ \\
$diag(\mathbf{h_l})$=$\begin{bmatrix}
h_{1,l}^2& 0 &\cdots & 0\\
0 & h_{2,l}^2 & \cdots & 0\\
\vdots & \vdots & \vdots & \vdots \\
0 & 0 & \cdots & h_{M,l}^2
\end{bmatrix}$ and $l \in \{d,1,2,\cdots,K\}$
  \section{Secrecy rate analysis}\label{sec:secrate}
  \subsection{Special Cases}
  \subsubsection{Symmetric Network}
  We consider a $M$ relay node symmetric network with single eavesdropper. By symmetric network we mean that the channel gain from source to all the relays are equal, \textit{i.e.} $h_{s,i}=h_{s,j}=h_s,\;\forall i,j \in \{1,2,\dots, M\}$. Same is true for the respective channel gains from any relay node to destination and relay node to eavesdropper; we denote them by $h_d$ and $h_e$, respectively. This kind of scenario can occur in real world if the relay nodes are co-located \textit{i.e.} relay nodes are quite close to each other or if there is a single relay with multiple antennas. For such a network model secrecy rate can be written as: 
  \begin{equation}\label{eq:gammaeq}
  R_s(\bm{\beta})=\frac{1}{2}\log\frac{\left( 1+\gamma'_s\frac{\left( \sum\limits_{i=1}^{M}\beta_i\right) ^2}{\nu+\sum\limits_{i=1}^{M}\beta_i^2}\right) }{\left( 1+\gamma'_s\frac{\left( \sum\limits_{i=1}^{M}\beta_i\right) ^2}{\mu+\sum\limits_{i=1}^{M}\beta_i^2}\right)}
  \end{equation}
  where $\gamma_s=\frac{P_s}{\sigma^2},\;\gamma'_s=h_s\gamma_s,\; \nu=\frac{1}{h_d^2},\;\mu=\frac{1}{h_e^2}$
  \begin{pavikp}
  For symmetric M relay node network optimum $\beta$ values are all equal, \textit{i.e.} $\beta^*_1=\beta^*_2=\dots=\beta^*_M$. 
  \end{pavikp}
  \begin{proof}
  By taking derivative of $R_s$ with respect to $\beta_i$, we get\begin{multline*}
  2\gamma'_s\Big( \sum\limits_{j=1}^{M}\beta_i\Big) \Big( \mu-\nu\Big) \Bigg\{ \Big( \sum\limits_{i=1}^{M}\beta_i^2-\beta_i\Big( \sum\limits_{j=1}^{M}\beta_j\Big) \Big)
  \Big( 2 \sum\limits_{i=1}^{M}\beta_i^2+\nu+\mu\Big)
  -\Big( \sum\limits_{i=1}^{M}\beta_i^2\Big) ^2-\gamma'_s\beta_i\Big( \sum\limits_{j=1}^{M}\beta_i\Big)^3+\nu\mu\Bigg\} =0
  \end{multline*}
  Now, as $\left( \sum\limits_{j=1}^{M}\beta_i\right)=0$ is not the optimal solution and $\mu\ne\nu$, hence
  \begin{multline}\label{gammaeq1}
  \left( \sum\limits_{i=1}^{M}\beta_i^2-\beta_i\left( \sum\limits_{j=1}^{M}\beta_j\right) \right)\left( 2 \sum\limits_{i=1}^{M}\beta_i^2+\nu+\mu\right) -\left( \sum\limits_{i=1}^{M}\beta_i^2\right) ^2
  -\gamma'_s\beta_i\left( \sum\limits_{j=1}^{M}\beta_i\right)^3+\nu\mu =0
  \end{multline}
  Similar equation can be obtained if we take derivative with respect to $\beta_j,\;j\ne i \text{ and } j \in \{1,2,\dots,M\}$. Subtracting this equation from equation \eqref{gammaeq1} we get
  {\small \begin{equation}
  (\beta_i-\beta_j)\left( \sum\limits_{j=1}^{M}\beta_j\right)\left\lbrace\left( 2 \sum\limits_{i=1}^{M}\beta_i^2+\nu+\mu\right)+\gamma'_s\left( \sum\limits_{j=1}^{M}\beta_i\right)^3\right\rbrace=0 
  \end{equation}}
  Hence, $\beta_i=\beta_j,\;i,j\in\{1,2,\dots,M\}$
  \end{proof}
  \begin{pavikc}
  Optimal solution of equation \eqref{eq:gammaeq} can be attained by solving the following optimization problem
  \begin{equation*}
  \max_{\beta}\frac{\left( 1+\frac{P_s}{\sigma^2} \frac{M^2\beta^2 h_s^2 h_d^2}{1 + M\beta^2 h_d^2}\right) }{\left( 1+\frac{P_s}{\sigma^2} \frac{M^2\beta^2 h_s^2 h_e^2}{1 + M\beta^2 h_e^2}\right) }
  \end{equation*} 
  \begin{equation*}
  \beta^*_{AF} = \begin{cases}
            \bigg[\frac{\sigma^2}{P_R M^2h_d^2 h_r^2}\bigg]^{1/4} \beta_{max}^{1/2}, \quad \frac{\sigma^2}{P_R \beta_{max}^2} < M^2h_d^2 h_e^2 \\
            \beta_{max}, \quad \mbox{o.w.}
            \end{cases}
  \end{equation*}
  \end{pavikc}
  \begin{proof}
  As the optimal solution corresponds to equal values of $\beta_i$'s, so replacing them by $\beta$ in equation \eqref{eq:gammaeq} we get the above expression. Now it is easy to see that if we introduce the following parameters: $h_{sr}=\sqrt{M}h_s$, $h_{r,d}=\sqrt{M}h_d$ and $h_{r,e}=\sqrt{M}h_e$ then the above problem become single relay secrecy rate maximization problem. So, we can use the solution of that problem which can be calculated easily.
  % calculate the optimal $\beta$ values for that problem.
  \end{proof}
  \subsection{Asymmetric Network with $\mathbf{h_e}=\alpha \mathbf{h_d}$}
  For sake of analysis we consider a new variable $\omega_i=h_{i,d}\beta_{i}$, its upper bound $\omega_{i,max}=h_{i,d}\beta_{i,max}$ and a new parameter $g_{s,i}=\sqrt{\frac{P_s}{\sigma^2}}h_{s,i},\;\forall i \in \{1,2,\cdots,M\}$. Therefore, secrecy rate expression can be written as:
  \begin{align*}
  R_s=\frac{1}{2}\log\frac{\left( 1+\frac{(\sum\limits_{i=1}^{M}g_{s,i}\omega_i)^2 }{1+\sum\limits_{i=1}^{M}\omega_i^2}\right)}{\left( 1+\frac{\alpha^2(\sum\limits_{i=1}^{M}g_{s,i}\omega_i)^2}{1+\alpha^2(\sum\limits_{i=1}^{M}\omega_i^2)}\right) }
  \end{align*}
   It is easy to see that non-zero secrecy rate can be obtained only if $\alpha <1$. If we denote $\sum\limits_{i=1}^{2}\omega_i^2=r^2$ then the above expression can be written as: $\frac{1+\varrho_1\Psi(\bm{\omega})}{1+\varrho_2\Psi(\bm{\omega})}$, %$\frac{1}{1+r^2} > \frac{\alpha^2}{1+\alpha^2r^2}$, 
  where $\Psi(\mathbf{\omega})=(\sum\limits_{i=1}^2 g_{s,i}\omega_i)^2$, $\varrho_1=\frac{1}{1+r^2}$ and $\varrho_2=\frac{\alpha^2}{1+\alpha^2r^2}$.
  It is easy to see that as $\alpha<1$, so $\varrho_1 > \varrho_2$ and therefore the expression $\frac{1+\varrho_1\Psi(\bm{\omega})}{1+\varrho_2\Psi(\bm{\omega})}$ will be maximized when $\Psi(\bm{\omega})$ is maximum.
  \par\textit{Individual Constraint:} To maximize $\Psi(\bm{\omega})$ we formulate following optimization problem.
  \begin{align*}
  \max \quad & \left( \sum\limits_{i=1}^M g_{s,i}\omega_i\right) ^2\\
  \text{s.t.} & \sum\limits_{i=1}^{M}\omega_i^2=r^2
  \end{align*} 
  One can easily calculate the optimal $\bm{\omega}$ vector which is equal to $\frac{\mathbf{g_s}}{||\mathbf{g_s}||}r$.
  By replacing this in the secrecy rate equation along with the sum constraint we can formulate following optimization problem in $r$:
   \begin{align*}
  \max_r \quad\frac{1+\frac{||\mathbf{g_s}||^2r^2}{1+r^2}}{1+\frac{\alpha^2||\mathbf{g_s}||^2r^2}{1+\alpha^2r^2}}
  \end{align*} 
  The solution to this problem is $r^*=\frac{1}{\sqrt{\alpha\sqrt{1+||\mathbf{g_s}||^2}}}$.
  Now, if  $\frac{\mathbf{g_s}}{||\mathbf{g_s}||}r^*$
   satisfy individual constraints of $\omega_i$ then this is the optimal solution, else we move to next step.
  This case considers the scenario when the solution obtained using above method violates any of the $\omega_i$'s constraints.\par We arrange $\omega_i$'s according to $\frac{g_{s,i}}{\omega_{i,max}}$ and denote them as $\omega_{(1)} \ge \omega_{(2)} \ge \dots \omega_{(M)}$. We follow the same ordering for $g_{s,i}$'s also and denote the ordered values as $g_{s,(i)}$. 
      Let us assume that the individual constraint is violated upto $m^{th}$ ordered variable. As the $m$ variables have violated their constraints we can replace them by their individual upper bound, whereas, for the rest of the variables we still need to find their optimum values. Now we introduce some notations for the analysis in the next section.
      \begin{align*}
      p_m &=\sum\limits_{i=1}^{m}g_{s,(i)}\omega_{(i),max}\\
      q_m &=\sum\limits_{i=1}^{m}\omega_{(i),max}^2\\
      s_m &=\sum\limits_{i=m+1}^{M}g_{s,(i)}^2
      \end{align*}
      Following the argument presented in \cite{jing} we can
      easily upper bound those variables.     
      
      {\small\begin{align*}
      \mathbf{g}_{s,\{(m+1),\cdot,(M)\}}^t\bm{\omega}_{\{(m+1),\cdot,(M)\}}&\le\mathbf{g}_{s,\{(m+1),\cdot,(M)\}}^t\lambda_m\mathbf{g}_{s,\{(m+1),\cdot,(M)\}}\\
      \text{for all }||\bm{\omega}_{\{(m+1),\cdot,(M)\}}||&=\sqrt{r^2-q_m}
      \end{align*}
          }
          where $\mathbf{u}_{\{(m+1),\cdot,M\}}=[u_{(m+1)},\cdots,u_{(M)}]$ and $\mathbf{u}\in \{\mathbf{g}_s,\bm{\omega}\}$.
       In simple words the objective function will be maximized when the vector composed of those variables lies in the direction of $[g_{s,(m+1)},\cdots,g_{s,(M)}]^t$ vector. Therefore, we have to find the optimal scaling factor ($\lambda_m$) for those variables. Now by replacing the first $m$ ordered variables by their respective upper bounds and rest of them by a scaled vector of $\mathbf{g}_{s,\{(m+1),\cdot,(M)\}}$ we obtain following optimization problem in terms of $\lambda_m$.\begin{align*}
           \max_{\lambda_m} \quad \frac{1+\frac{(p_m+s_m\lambda_m)^2}{1+q_m+s_m\lambda_m^2}}{1+\frac{\alpha^2(p_m+s_m\lambda_m)^2}{1+\alpha^2q_m+\alpha^2s_m\lambda_m^2}}    
           \end{align*}
           %where $\gamma_s=\frac{P_s}{\sigma^2}$. 
           The solution to the above problem can be calculated by solving the following $4^{th}$ degree polynomial $P(\lambda_m)$.
          
       \begin{multline*}
       \lambda_m^4+\frac{p_m(2+3\gamma_ss_m)}{s_m(1+\gamma_ss_m)}\lambda_m^3+\frac{3p_m^2\gamma_s}{s_m(1+\gamma_ss_m)}\lambda_m^2+\\ \frac{p_m(\gamma_sp_m^2\alpha^2+2q_m\alpha^2+\alpha^2+1)}{s^2\alpha^2(1+\gamma_ss_m)}\lambda_m-\frac{(1+q_m)(1+\alpha^2q_m)}{s^2\alpha^2(1+\gamma_ss_m)}=0
       \end{multline*}
       
           \textit{Characterization of the roots}: As we can see that the co-efficients of the polynomial $P(\lambda_m)$ are all positive except the last one and therefore there is only one variation in sign of co-efficients. Using the \textit{Descartes' rule of sign} the number of positive root has to be one.\\
           \textit{Discussion:} Based on the number of variables violating their individual constraints we can divide the feasible set in several non-overlapping subsets. For example, let us assume that the solution obtained using first approach has only the first variable (after ordering) at its boundary. This is a feasible point for individual constraints also. We denote the corresponding $r$ as $r_1$. We can write the optimal solution in terms of $r_1$ in following way: $\bm{\omega}^*=\frac{\mathbf{g_s}}{||\mathbf{g_s}||}r_1=\frac{\mathbf{h_s}}{||\mathbf{h_s}||}r_1$ and for the first variable $\omega_{(1),max}=\frac{h_{s,(1)}}{||\mathbf{h_s}||}r_1 \implies r_1 = \frac{||\mathbf{h_s}||}{h_{s,(1)}}\omega_{(1),max}=\sqrt{\frac{s_1}{h_{s,(1)}^2}\omega_{(1),max}^2+q_1}$. Similarly,   
           %\par Following the approach mentioned above 
           we can write $r_m=\sqrt{\frac{s_m}{h_{s,(m)}^2}\omega_{(m),max}^2+q_m}$ for $m \in \{1,2,\dots,(M-1)\}$ and $r_M=q_M$. Now, if we consider $r \in [r_m,r_{m+1}]$ region then the start point, i.e. $r_m$ indicates that $m$ ordered variables have reached their corresponding boundary and scaling factor $\lambda=\frac{\omega_{(m),max}}{h_{s,(m)}}$. On the otherhand at $r_{m+1}$, $(m+1)^{th}$ variable has just reached its boundary, i.e. $\lambda=\frac{\omega_{(m+1),max}}{h_{s,(m+1)}}$.\\
           \textit{Remark:} If $0 \le r \le \sqrt{\frac{s_1}{h_{s,(1)}^2}\omega_{(1),max}^2+q_1}$, then we can obtain the optimal scaling vector using first approach.
           
\subsection{General Case with degraded eavesdroppers channels}
In this subsection we evaluate the optimal $\bm{\beta}$ vector for two different kind of constraints without imposing any assumption on channel gains. At first we consider the \textit{zero forcing} solution where $\beta$ values are chosen such that the transmitted signal get canceled at eavesdropper. We formulate a quadratic program with individual constraint for this scenario. In the subsequent paragraph we formulate an optimization problem with total sum constraint on $\bm{\beta}$ vector and discuss an iterative approach for calculating the optimal value. It is easy to see that the zero forcing solution will lower bound the optimum secrecy rate for individual constraint whereas total sum constraint will upper bound the same.\\
\textbf{Zero Forcing for individual constraints:} 
In this approach we equate the $SNR_e$ to zero. As a result the equivalent optimization problem can be written as:
           \begin{subequations}\label{zfs}
           \begin{align}
           \max_{\bm{\beta}}&\quad\frac{\left(  \sum\limits_{i=1}^{M} h_{s,i}\beta_ih_{i,d}\right)  ^2}{1 + \sum\limits_{i=1}^{M}(\beta_i h_{i,d})^2}\label{zfs_obj}\\
           \text{such that}&\;  \sum\limits_{i=1}^{M} h_{s,i}\beta_ih_{i,k}=0,\; k \in\{1,2,\dots,K\}\label{zfs_eve}\\
           & -\beta_{max,i} \le \beta_i \le \beta_{max,i}, \; i \in \{1,2,\dots M \}\label{zfs_rel}
           \end{align}
           \end{subequations}
           We can also formulate following quadratic program which can be solved efficiently. 
           \begin{pavikp}
           The optimization problem \ref{zfs} is equivalent to following quadratic program
           \begin{align*}
           \max_{\mathbf{w}}\quad& \mathbf{w^tw}\\
           \text{s.t.} \quad %& \mathbf{\tilde{h}_s^tz}=1\\
           %&\mathbf{\tilde{h}_{s\rho}^tz}=0\\
           & \tilde{\mathbf{H}}_\rho\mathbf{w}=\begin{bmatrix}
           1\\
           \mathbf{0}
           \end{bmatrix} \\
           & \mathbf{H}_\beta\mathbf{w} \le \mathbf{0}
           \end{align*}
           \end{pavikp}
           \begin{proof}
           It can be shown that the optimal solution does not change if we rewrite the objective function as $\min\;\frac{1 +\sum\limits_{i=1}^{M}(\beta_i h_{i,d})^2}{\left( \sum\limits_{i=1}^{M} h_{s,i}\beta_ih_{i,d}\right)^2}$, because $ \sum\limits_{i=1}^{M} h_{s,i}\beta_ih_{i,d}\ne 0$. If we consider a new variable $v$ such that $\sum\limits_{i=1}^{M} h_{s,i}\beta_ih_{i,d}=v$, then in terms of variable vector $\mathbf{z}=[\frac{\bm{\omega}}{v}^{\mathbf{t}},\frac{1}{v}]$, we can easily write a quadratic optimization problem, where $\omega_i=h_{i,d}\beta_i,\;\forall i \in \{1,2,\dots,M\}$. The denominator of the objective function \eqref{zfs_obj} can be written as a constraint in terms of this variable vector as, $[\mathbf{h}_s^\mathbf{t},0]\mathbf{w}=1$.  We consider two new parameters $\rho_{i,k}=\frac{h_{i,k}}{h_{i,d}},\;\forall i$ and $h_{sik}=h_{s,i}\rho_{i,k},\;\forall i$. We then define matrix $\mathbf{H}_\rho$ such that 
           $(\mathbf{H}_\rho)_{k,i}=h_{sik}, i \in \{1,2,\cdots,M\} \text { and } \forall k $ and $(\mathbf{H}_\rho)_{k,M+1}=0,\;\forall k$. 
            %$\mathbf{\tilde{h}}_s=[\mathbf{h}_s,0]$ and $\mathbf{H}_\rho\mathbf{z}=0$, where  $\mathbf{\tilde{h}}_{s\rho}=[\mathbf{h}_{s\rho},0]$ and
            $\tilde{\mathbf{H}}_\rho$ is generated by concatenating the constraint due to denominator and eavesdroppers.
             $\mathbf{H}_\beta$ can be obtained by rearranging the following constraint: $-h_{d,i}\beta_{i,max}w_{M+1} \le w_i \le h_{d,i}\beta_{i,max}w_{M+1},\forall i$ 
           \end{proof}
\textbf{Optimal Secrecy rate for total Sum Constraint: }
From equation \eqref{eq:opt1} we can rewrite the equivalent optimization problem in following manner:
\begin{align*}
\max_\eta\;\;\max_{\bm{\beta}}\;\;& \frac{1+SNR_d(\bm{\beta})}{1+\eta}\\
\text{s.t.}\;\; & SNR_k(\bm{\beta})\le \eta,\;k \in \{1,2,\dots,K\}\\
& \bm{\beta}^t\bm{\beta} \le \beta_{tot} %\sum\limits_{i=1}^{M}\beta_{i,max}^2
\end{align*}

\textit{Our Approach}: This problem becomes a single dimensional optimization problem if we know the solution of inner optimization problem for a fixed $\eta$. Then we can search over the range of $\eta$ for the optimum $\eta$ at which the maximum objective function value is attained. We can use bi-section algorithm to find the optimal $\eta$ and the corresponding optimum decision variable $\bm{\beta}$.

\par \textit{Range of $\eta$:} As $\eta=0$ will result in sub-optimal zero-forcing solution,
% \textbf{zero forcing solution} and we have a different approach for that. 
so we can start with small values of $\eta$, typically in the range of $10^{-6}$. The upper bound on $\eta$ can be calculated by solving the following optimization problem:
\begin{subequations}\label{eta_range}
\begin{align}
\max_{\bm{\beta}}&\;\frac{(\mathbf{h_{s,k}^t}\bm{\beta})^2}{1+\bm{\beta}^t\mathbf{D_k}\bm{\beta}}\\
\text{s.t.}&\;\bm{\beta}^t\bm{\beta} \le \beta_{tot} %\sum\limits_{i=1}^{M}\beta_{i,max}^2
\end{align}
\end{subequations}
$\mathbf{h_{s,k}}=[h_{s,1}h_{1,k},h_{s,2}h_{2,k},\cdots,h_{s,M}h_{M,k}]^\mathbf{t}$, $\mathbf{D_k}=diag(h_{1,k}^2,h_{2,k}^2,\cdots, h_{M,k}^2) \quad \forall k \in \{1,2,\cdots,K\}$. The solution of this problem is $\eta_{k,max}=\mathbf{h_{s,k}}^t\mathbf{\tilde{D}_k}^{-1}\mathbf{h_{s,k}}$, where
$\mathbf{\tilde{D}_k}=\frac{1}{\beta_{tot}}\mathbf{I}+\mathbf{D_k}$ %$\mathbf{\tilde{D}_k}=\frac{1}{\sum\limits_{i=1}^{M}\beta_{i,max}^2}\mathbf{I}+\mathbf{D_k}$
 and $\eta_{max}=\underset{k}{\max}\quad \eta_{k,max}$.
 \par We consider the following transformation for converting the non-convex problem into a convex one.
 \begin{subequations}\label{eq:trnsfrm}
 \begin{align}
 \omega_i=\beta_ih_{i,d}&,\;\text{and } v_i=\frac{\omega_i}{\sqrt{1+\sum\limits_{i=1}^{M}\omega_i^2}}\\
 \text{or compactly }& \mathbf{v}=\frac{\bm{\omega}}{\sqrt{1+\bm{\omega^t\omega}}} \Leftrightarrow \bm{\omega}=\frac{\mathbf{v}}{\sqrt{1-\mathbf{v^tv}}}
 \end{align}
 \end{subequations}
 %Therefore the original objective function can be written as
 \begin{itemize}
 \item \textbf{Objective function}: For a fixed $\eta$ our objective is to maximize $(\sum\limits_{i=1}^{M}h_{s,i}v_i)^2$ or equivalently $\mathbf{v^th_sh_s^tv}$.
 \item \textbf{Eavesdroppers constraint}: From equation \eqref{eq:snr_vec} eavesdroppers' SNR constraints can be written as:
 \begin{align*}\label{matP}
  SNR_k=&\frac{(\mathbf{h_{s,k}^t}\bm{\beta})^2}{1+\bm{\beta}^\mathbf{t}diag(\mathbf{h_k})\bm{\beta}}\frac{P_s}{\sigma^2} \le \eta,\; \forall k \in \{1,2,\cdots,K\}
  \end{align*}
  Using the same transformation on the SNR constraint due to eavesdropper we get:
 \begin{align*}
 \frac{\bm{\omega}^t(\mathbf{h}_{s\rho,k}\mathbf{h}_{s\rho,k}^t)\bm{\omega}}{1+\bm{\omega}^t\mathbf{D}_{\rho k}\bm{\omega}}\frac{P_s}{\sigma^2} \le \eta
 \end{align*}
 where $\mathbf{h}_{s\rho k}=[h_{s,1}\rho_{1,k},h_{s,2}\rho_{2,k},\cdots,h_{s,M}\rho_{M,k}]^\mathbf{t}, \mathbf{D}_{\rho k}=diag([\rho_{1,k}^2,\rho_{2,k}^2,\cdots,\rho_{M,k}^2])$ and $\rho_{i,k}=\frac{h_{i,k}}{h_{i,d}},\;\forall i$ 
 We transform the above expression in terms of $\mathbf{v}$ and rearrange it to obtain:
 \begin{align*}
   \mathbf{v^t}\mathbf{C}_k(\eta)\mathbf{v} \le 1 \text{ where }
   \mathbf{C}_k(\eta)=&\frac{\mathbf{h}_{s\rho k}\mathbf{h}_{s\rho k}^\mathbf{t}}{\eta}\frac{P_s}{\sigma^2}+\mathbf{I}-\mathbf{D}_{\rho k}\quad 
 \end{align*}
  If $\rho_{i,k}=\frac{h_{i,k}}{h_{i,d}},\;\forall i,k$ then $\mathbf{I}-\mathbf{D}_{\rho k}$ is a diagonal matrix with positive entries, therefore, $\mathbf{C}_k$ is a positive definite matrix.
 \item \textbf{Total Constraint}: Total constraint can be written in the following vector notations:
 \begin{align*}
 &\sum\limits_{i=1}^{M}\beta_i^2 \le \beta_{tot}\\ %\sum\limits_{i=1}^{M}\beta_{i,max}^2\\
 \implies & \sum\limits_{i=1}^{M}\frac{\omega_i^2}{h_{i,d}^2}  \le \beta_{tot} \\%\sum\limits_{i=1}^{M}\beta_{i,max}^2\\
 %\mathbf{y^ty}\le&\sum\limits_{i=1}^{M}y_{i,max}^2=:\upsilon\\
 \implies & \sum\limits_{i=1}^{M}\frac{v_i^2}{h_{i,d}^2(1-\sum\limits_{i=1}^{M}v_i^2)}\le \beta_{tot} \\ %\sum\limits_{i=1}^{M}\beta_{i,max}^2=:\kappa\\
 \implies & \sum\limits_{i=1}^{M}(1+\frac{1}{h_{i,d}^2\beta_{tot}})v_i^2 \le 1\\
 \text{ or }& \mathbf{v^tD_Tv}\le 1 \text{ where } \mathbf{D_T}=diag([(1+\frac{1}{h_{i,d}^2\beta_{tot}}),\forall i ]) %\mathbf{x^tx}\le\frac{\upsilon}{1+\upsilon}=:\tau
 \end{align*}
 \end{itemize}
 As $h_{s,i}>0,\;\forall i$ and the constraints are quadratic in nature, so we claim that we can replace the quadratic objective function by a linear one \textit{i.e.}, $\mathbf{h_s^tv}$. Though maximum value of $(\mathbf{h_s^tv})^2$ for the given constraints can be obtained by finding the maximum and minimum value of linear objective $\mathbf{h_s^tv}$ for those same constraints, but those values will be indeed same. We can argue that using contradiction. Let us assume that the solutions of maximization and minimization problem are $\widehat{\mathbf{v}}^*$ and $\widetilde{\mathbf{v}}^*$, respectively. Now, if $\mathbf{h_s^t\widehat{v}^*}<|\mathbf{h_s^t\widetilde{v}^*}|$, then we can find a vector $-\widetilde{\mathbf{v}}^*$ which will not only satisfy the constraints but also has higher objective function value than $\widehat{\mathbf{v}}^*$. Hence, $\widehat{\mathbf{v}}^*$ cannot be optimum. Similar argument can be presented for minimization problem also and thus it proves our claim.
 \par For a fixed $\eta$ the reformulated optimization problem becomes:
 \begin{subequations}\label{opt}
 \begin{align}
 %\max_\eta\;\;
 \max_{\mathbf{x}}\;\;& \mathbf{h_s^tv}\\
 \text{s.t.}\qquad & \mathbf{v^t}\mathbf{C}_k(\eta)\mathbf{v} \le 1,\;k \in \{1,2,\dots K\}\\
 & \mathbf{v^tD_Tv}\le 1
 \end{align}
 \end{subequations}
 Now, it is easy to see that $\mathbf{v^tD_Tv} \le 1$ is a M dimensional ellipsoid, also, $\mathbf{C}_k(\eta),\;k \in \{1,2,\dots,K\}$ are positive semidefinite symmetric matrix.
 Hence, this is a convex optimization problem and
 %Quadratic Constrained Quadratic Program (QCQP) and
  therefore global solution can be obtained using numerical routines. Once the optimal solution for a particular $\eta$ is obtained, then we can calculate corresponding secrecy rate by evaluating $\frac{1+(\mathbf{h_s^tv}^*(\eta))^2}{1+\eta}$. Now, we use any line search method to calculate optimal $\eta^*$ due to following proposition. We have used golden-section search for generating the results.
 \begin{pavikp}
 $\frac{1+(\mathbf{h_s^tv}^*(\eta))^2}{1+\eta}$ is an unimodal function of $\eta$ in the range $[0, \infty)$%\eta_{max}]$.
 \end{pavikp}
 \begin{proof}
For lower values of $\eta$, eavesdroppers' constraints of optimization problem \eqref{opt} are dominating and as $\eta$ increases, so the volume of ellipsoids corresponding to those constraints. This results in enlargement of the feasible region, which causes increment in the objective function value of problem \eqref{opt}. Therefore, for values around $\eta=0$, $\frac{1+(\mathbf{h_s^tv}^*(\eta))^2}{1+\eta}$ increases with $\eta$. For higher values of $\eta$, $\mathbf{v^tD_Tv}\le 1$ is dominating constraint and objective function value become constant for those values of $\eta$. So, as $\eta$ increases objective function starts decreasing. Hence, there is an intermediate point where $\frac{1+(\mathbf{h_s^tv}^*(\eta))^2}{1+\eta}$ reaches maximum value. 
 \end{proof}
% Now, it is easy to see that $\mathbf{x^tD_Tx}= 1$ is a M dimensional ellipsoid. Also, as $\mathbf{C}(\eta)$ is a symmetric matrix all its eigen values will be real and if they are all positive then it is indeed an ellipsoid. For a symmetric
% % and \textcolor{red}{invertible} 
%  matrix we can express it in following way: $\mathbf{C}=\mathbf{UDU}^t$, where $\mathbf{UU}^t=\mathbf{I}$.
In the following subsection we discuss the single eavesdropper scenario with sum relay constraint and characterize the solution of secrecy rate maximization problem. 
\subsection*{\textbf{Characterization of Solution for single eavesdropper scenario}}
 \begin{itemize}
 \item \textbf{Case 1}: If $\lambda_{min}(\mathbf{C})$ is the minimum eigen value of matrix $\mathbf{C}$ or $\mathbf{D}$ then as long the maximum eigen value of $\lambda_{max}(\mathbf{D_T}) < \lambda_{min}(\mathbf{C})$, the constraint $\mathbf{x^tD_Tx}\le 1$ will be inactive and hence we are left with maximization of quadratic objective with a quadratic equality constraint. This is indeed generalized Rayleigh-quotient \cite[p.~176]{horn1985} and the solution of this problem can be easily obtained by calculating the eigen vector corresponding to maximum eigen value of the matrix $\mathbf{C}(\eta)^{-1}\mathbf{h_sh_s}^t$. To calculate the eigen value of the matrix we can solve the following equation:
 \begin{align*}
 \mathbf{C}(\eta)^{-1}\mathbf{h_sh_s^t}\mathbf{v}=\lambda\mathbf{v}
 \end{align*}
 As $\mathbf{h_s}^t\mathbf{v}$ is a scalar, then $\mathbf{C}(\eta)^{-1}\mathbf{h_s}$ indeed lies in the direction of $\mathbf{v}$ and infact by neglecting scale factor we can write $\mathbf{v}=\mathbf{C}(\eta)^{-1}\mathbf{h_s}$.
 %As the second matrix is rank-1, hence their product will also be rank-1 and the eigen vector corresponds to non-zero eigen value.  
% using the following lemma we can solve the above optimization problem for that fixed $\eta$.  
% \begin{pavikl}\label{lemma1}
% The solution of following maximization problem 
% \begin{align*}
% \max \quad & \mathbf{w}^T\mathbf{x} \\
% \text{s.t.} \quad & \mathbf{x}^T\mathbf{Ax}= 1
% %& 0 \le \mathbf{x}_{mn} \le \mathbf{x} \le \mathbf{x}_{mx} %& x_{2,mn} \le x_2 \le x_{2,mx}
% \end{align*}
% is $\mathbf{x}^*=\frac{1}{\sqrt{\mathbf{w^TA}^{-1}\mathbf{w}}}\mathbf{A}^{-1}\mathbf{w}$.
% \end{pavikl}
% This approach provides us the optimal $\mathbf{x}$ vector and by replacing it in objective function we get:
% %So we can write our optimization problem as 
% \[\mathbf{h_s^Tx}=\frac{\mathbf{h_s}^T\mathbf{P}(\eta)^{-1}\mathbf{h_s}}{\sqrt{\mathbf{h_s^TP}(\eta)^{-1}\mathbf{h_s}}}= \sqrt{\mathbf{h_s^TP}(\eta)^{-1}\mathbf{h_s}} \]
 Therefore, for this fixed $\eta$ the secrecy objective function becomes:
 \[R_s(\eta)=\frac{1+\mathbf{h_s^tC}(\eta)^{-1}\mathbf{h_s}}{1+\eta}\]
 \item \textbf{Case 2}: When the criteria mentioned above is not satisfied then both the constraints might be active, so we use following two step approach:
 \begin{itemize}
 \item We solve the problem considering the first constraint only. If the solution obtained ($\mathbf{x}^*$) satisfy the second constraint then this is the solution else we discard it and follow the next step.
 \item In this case both the constraints are active and we have to solve the following problem.
 \begin{subequations}
 \begin{align}
 %\max_\eta\;\;
 \max_{\mathbf{x}}\;\;& \mathbf{x^th_sh_s^tx}\\
 \text{s.t.}\qquad & \mathbf{x}^T\mathbf{C}(\eta)\mathbf{x} = 1\\
 & \mathbf{x^tD_Tx}= 1
 \end{align}
 \end{subequations}
 \end{itemize}
 \end{itemize}
 Now for this problem we can use any suitable numerical routine to find the global optimum which should satisfy the criteria mentioned in \cite{raey}.
 \subsection*{\textbf{Optimal Secrecy rate for individual relay constraints}} The objective function and eavesdroppers' constraints remain same as it was for total relay constraint scenario. But, unlike the previous case, instead of single relay constraint we have $M$ relay constraints corresponding to each of the relay nodes. The individual relay constraint and the transformed one is presented below:
 : 
       \begin{align*}
       & \beta_i^2=\frac{\omega_i^2}{h_{i,d}^2}=\frac{v_i^2}{h_{i,d}^2(1-\mathbf{v^tv})} \le \beta_{i,max}^2,\;\forall i \in \{1,2,\cdots,M\}\\
       \implies & \mathbf{v^tv}+\frac{v_i^2}{h_{i,d}^2\beta_{i,max}^2} \le 1 \text{ or } \mathbf{v^tD_iv} \le 1\\
        \text{where }&(\mathbf{D_i})_{jk}=
       \begin{cases}
       & 1+ \frac{v_i^2}{h_{i,d}^2\beta_{i,max}^2},\;\text{if }k=j=i\\
       & 1 ,\text{ if }k=j\ne i\\
       & 0, \text{ otherwise } 
       \end{cases}
       \end{align*}
   We can upper bound the $\eta_{k,max},\;\forall k$ in this scenario by using the solution of optimization problem \eqref{eta_range}. We use $\beta_{tot}= \sum\limits_{i=1}^{M}\beta_{i,max}^2$ in this case to evaluate the upper bound of $\eta_{k,max}$, thereafter to be denoted as $\widehat{\eta}_{k,max},\; \forall k$. Similarly, $\widehat{\eta}_{\max}=\underset{k}{\max}\;\widehat{\eta}_{k,max}$. The inner optimization problem for a fixed $\eta$ can be written as:
   \begin{subequations}\label{opt_indv}
         \begin{align}
         \max_\mathbf{v}&\; \mathbf{h_s^tv}\\
         \text{such that}&\;\mathbf{v^tC_k}(\eta)\mathbf{v} \le 1,\;k \in \{1,2,\dots K \}\\
         & \mathbf{v^tD_iv} \le 1, \; i \in \{1,2,\dots M \}
         \end{align}
          \end{subequations}
 Similar argument for unimodularity of $\frac{1+(\mathbf{h_s^tv}^*(\eta))^2}{1+\eta}$ as function of  $\eta$ within the range $[0,\widehat{\eta}_{\max}]$ can be presented in this case also. Hence, by using golden-section search we can obtain optimal $\eta^*$ and thereby optimal $\bm{\beta}^*$ for individual relay constraint scenario with multiple eavesdroppers.
\par 
 \textbf{Convergence and Iterations:} The $f(\eta)=\frac{1+(\mathbf{h_s^Tx^*})^2}{1+\eta}$ is a continuous function of $\eta$ within the interval $[0,\eta_{max}]$, where $\mathbf{x^*}$ is the optimal solution of problem \eqref{opt}. Also, it is well known that for golden-section search method \cite[chap.~2.1]{burden1989numerical} after $n$ iterations the updated interval can be written as:
 \begin{align*}
 \eta_u^{(n)}-\eta_l^{(n)}=\tau^n\eta_{max}\;\text{ and } \eta^* \in [\eta_l^{(n)},\eta_u^{(n)}]
 \end{align*}
 where $\tau=0.618$.
 In other words  %$\eta^{(n)}=\frac{\eta_u^{(n)}+\eta_l^{(n)}}{2}$, so
  $(\eta_u^{(n)}-\eta^*) \le \tau^n\eta_{max}$ and $(\eta^* -\eta_l^{(n)}) \le \tau^n\eta_{max}$ and therefore the sequence $\{\eta_l^{(n)}\}_{n=1}^\infty$ and $\{\eta_u^{(n)}\}_{n=1}^\infty$ linearly converges to $\eta^*$ \cite{wilde1964optimum}. %with a rate of convergence. %$O(\frac{1}{2^n})$. 
 \par If we consider a tolerance parameter $\delta$, we can easily estimate the number of the iterations required.
 \begin{align*}
 \tau^N\eta_{max} \le \delta \implies N \ge 2.08\ln\frac{\eta_{max}}{\delta}
 \end{align*}
  \section{Results}\label{sec:res}
  For evaluation of secrecy rate we consider a network whose main channel gains are sampled from a Rayleigh distribution with parameter 0.5. To obtain the degraded channels for eavesdroppers, we multiply the relay to destination channel gains with the samples from \textit{Uniform distribution}$[0,1]$. We average the results of 100 such networks while plotting the graphs.
  In Figure \ref{fig_PsRsbeta} we plot the variation of optimal $\beta$ value and secrecy rate ($R_s$) with respect to source power ($P_s$) for symmetric network case. As the source power increases the bounds on $\beta$ value keep contracting and therefore the optimal $\beta$ value starts declining. This results in saturation of the secrecy rate. 
  \begin{figure}[!h]
  \centering
  \includegraphics[scale=0.65]{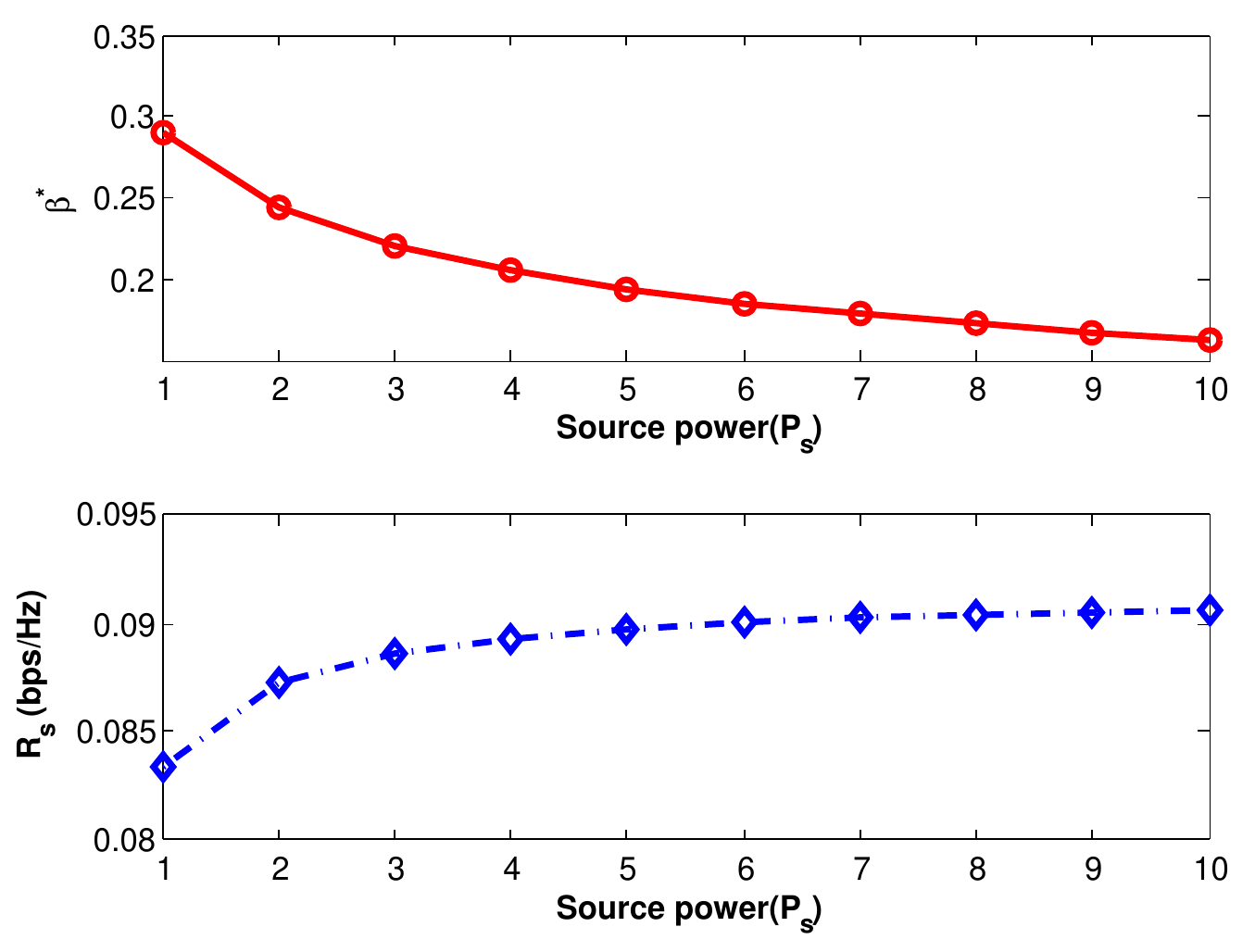}
  \caption{Plot of optimal $\beta$ and secrecy rate ($R_s$) with respect to $P_s$ for symmetric network case.}
  \label{fig_PsRsbeta}
  \end{figure}
  In figure \ref{fig_obsrv_bar_plot} we compare the secrecy rate obtained using proposed iterative approach with the optimal solution (solving \eqref{eq:secrate} directly using numerical routines) for randomly generated channel values. It is apparent that the outputs of the proposed iterative are equal to the corresponding optimal values. 
  \begin{figure}[!h]
  \centering
  \includegraphics[scale=0.65]{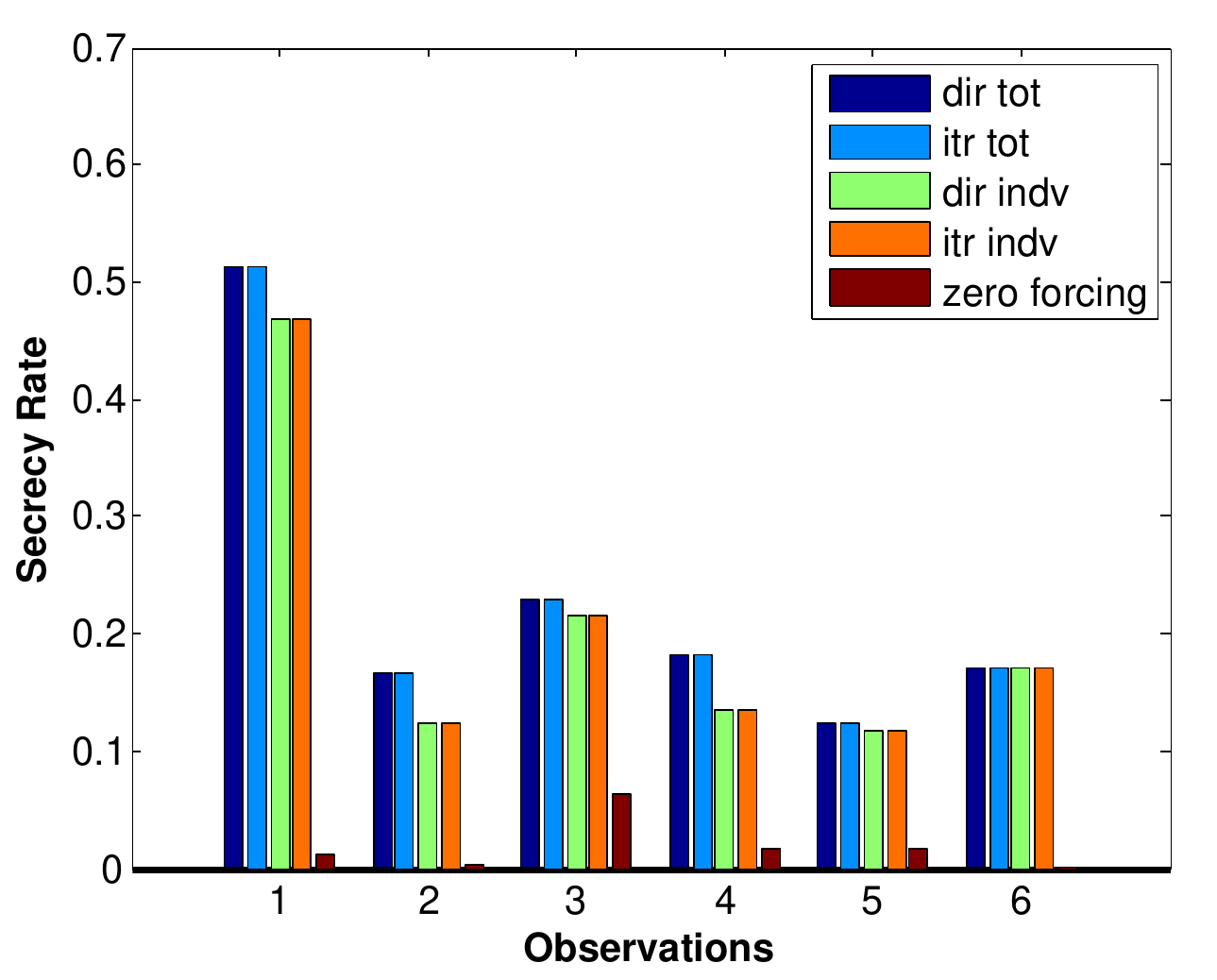}
  \caption{Comparison of secrecy rate ($R_s^*$) obtained for 5 relay node diamond network with 3 eavesdroppers using iterative approach, direct and zero-forcing formulation. Here we used $P_r=5$, $P_s=1$, $\sigma^2=1$. }
  \label{fig_obsrv_bar_plot}
  \end{figure}
  \begin{figure}[!h]
    \centering
    \includegraphics[scale=0.65]{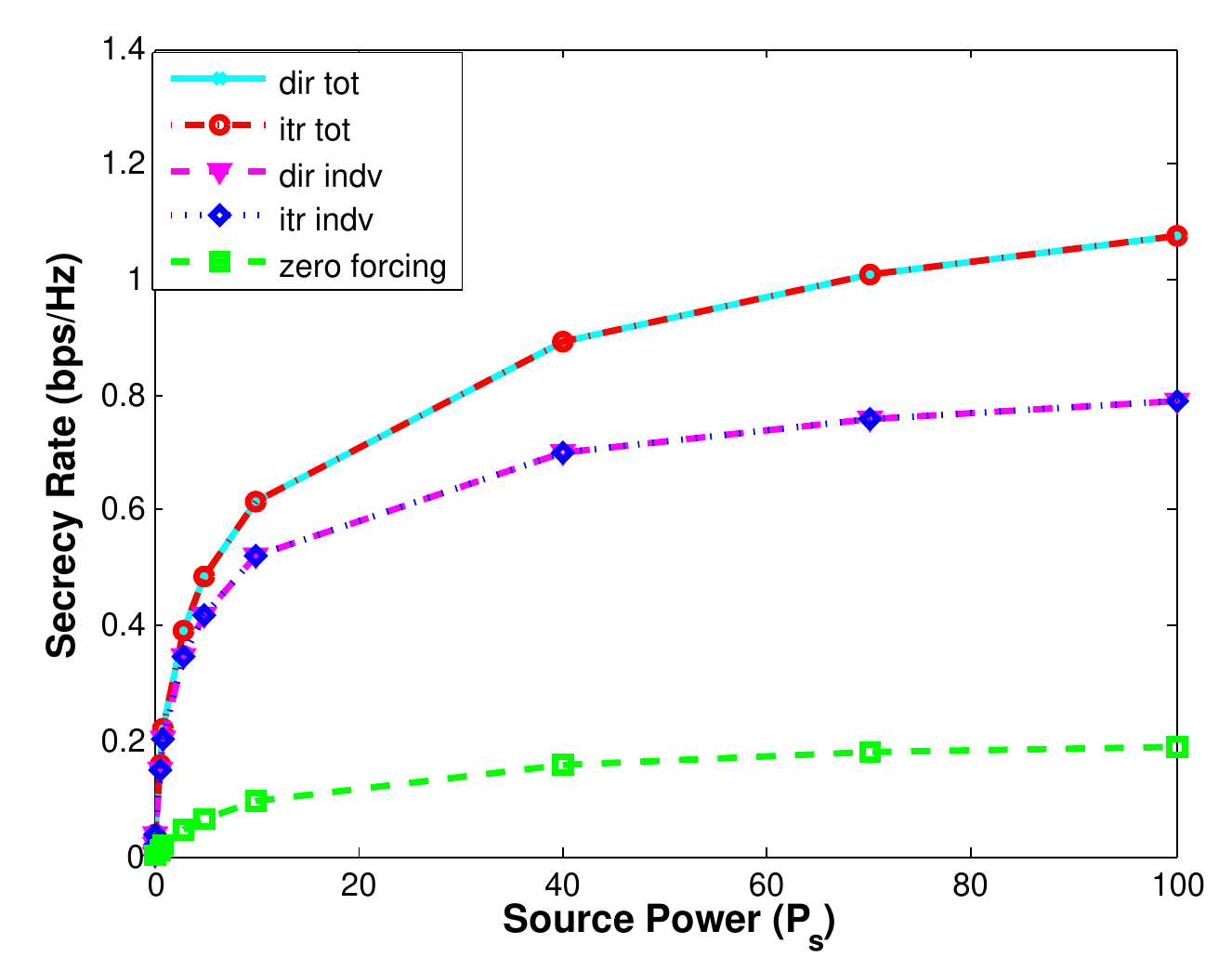}
    \caption{Plot of optimal secrecy rate ($R_s^*$) vs. source power ($P_s$) for five (5) relay node diamond network with three (3) eavesdroppers using iterative approach, direct and zero-forcing formulation. Here we considered $P_r=5$ and $\sigma^2=1$ }
    \label{fig_secrate}
    \end{figure}
    In Figure \ref{fig_secrate} we plot the secrecy rate  with respect to source power ($P_s$) for direct solutions and the solutions obtained using iterative algorithm and zero forcing approach. The reason behind the shape of these curves is already discussed in context of Figure \ref{fig_PsRsbeta}.
    \begin{figure}[!h]
       \centering
       \includegraphics[scale=0.65]{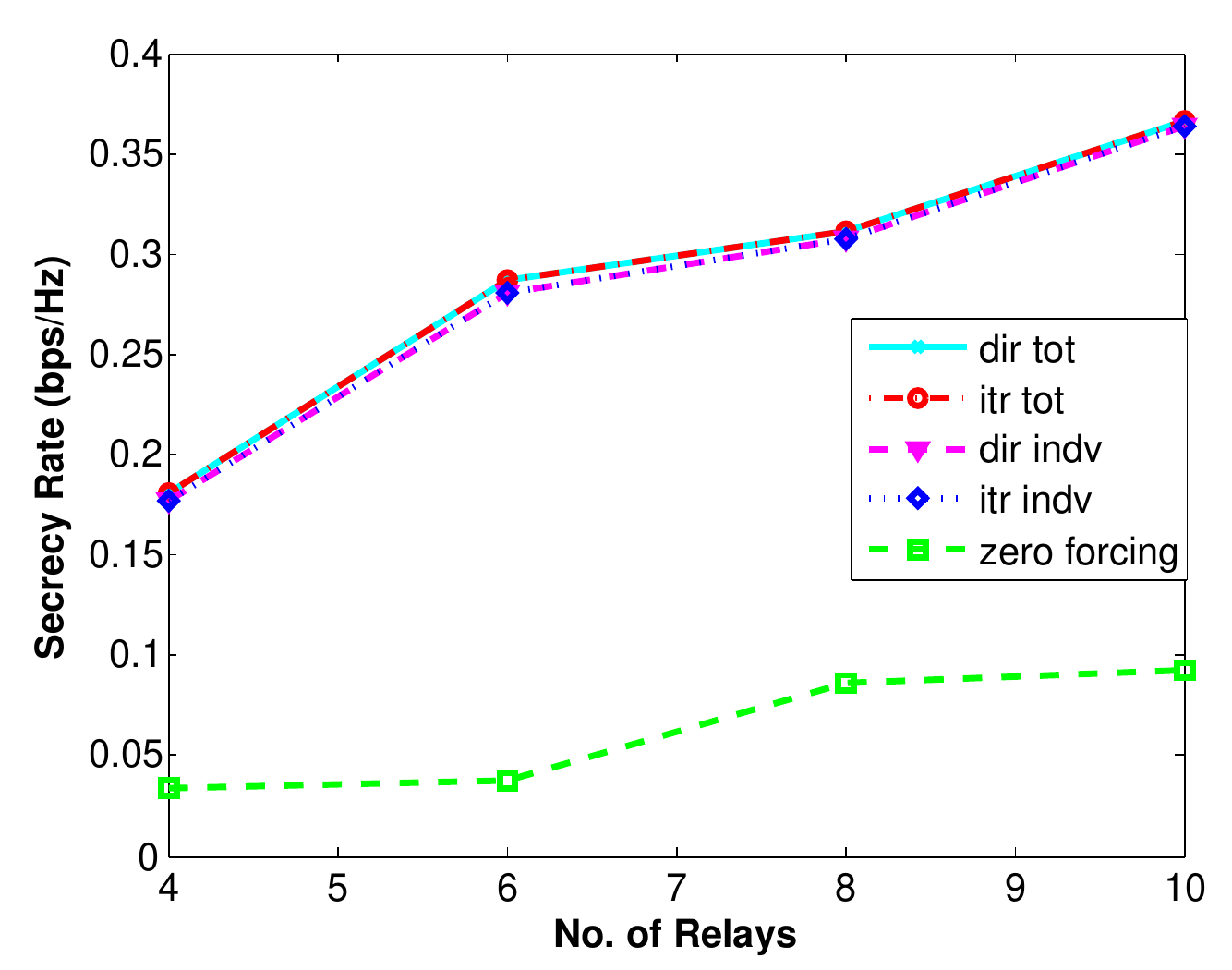}
       \caption{Plot of Secrecy rate ($R_s^*$) vs. No. of relay nodes ($M$) for diamond network with three (3) eavesdroppers for  iterative approach, direct and zero-forcing formulation. The parameters used are $P_r=5$, $P_s=1$, $\sigma^2=1$. }
       \label{fig_secvsrelays}
       \end{figure}
  Figure \ref{fig_secvsrelays} depicts variation of secrecy with respect to number of relay nodes deployed. As the number of relay node increases new paths from source to destinations are available and therefore, by choosing proper scaling factor we can achieve better secrecy rate. This applies for all three schemes (zero forcing, individual constraint and sum constraint) and is also evident from the plot.
 % The comparison among the results for randomly generated channel values for problem \textbf{P1} is depicted in figure \ref{fig_powmin}. ``Optimal" values are obtained by solving the problem \textbf{P1} directly using numerical approach, whereas ``SDP" values are output of the relaxed semidefinite program \textbf{P1-SDP}. 
 
  \section{Conclusion}\label{sec:concl}
  We have calculated the optimal scaling vector for two-hop amplify and forward (AF) to obtain the optimum secrecy rate in the presence of multiple eavesdroppers. We begin with considering the special channel conditions for the diamond network and gradually moved to general scenario. Analytical solution for special channel conditions and numerical solution for general scenario is proposed. In future we would like to investigate optimum secrecy rate of a general AF network. Also, as friendly jamming improves the secrecy rate in several scenarios, we would like to investigate the impact of jamming on secrecy rate in AF networks.


\begin{thebibliography}{10}
\providecommand{\url}[1]{#1}
\csname url@samestyle\endcsname
\providecommand{\newblock}{\relax}
\providecommand{\bibinfo}[2]{#2}
\providecommand{\BIBentrySTDinterwordspacing}{\spaceskip=0pt\relax}
\providecommand{\BIBentryALTinterwordstretchfactor}{4}
\providecommand{\BIBentryALTinterwordspacing}{\spaceskip=\fontdimen2\font plus
\BIBentryALTinterwordstretchfactor\fontdimen3\font minus
  \fontdimen4\font\relax}
\providecommand{\BIBforeignlanguage}[2]{{%
\expandafter\ifx\csname l@#1\endcsname\relax
\typeout{** WARNING: IEEEtran.bst: No hyphenation pattern has been}%
\typeout{** loaded for the language `#1'. Using the pattern for}%
\typeout{** the default language instead.}%
\else
\language=\csname l@#1\endcsname
\fi
#2}}
\providecommand{\BIBdecl}{\relax}
\BIBdecl
\bibitem{wyner}
A.~Wyner, ``{The Wire-tap Channel},'' \emph{Bell Systems Technical Journal},
  vol.~54, no.~8, pp. 1355--1387, Jan 1975.
\bibitem{schein2001phd}
B.~E. Schein, ``Distributed coordination in network information theory,'' Ph.D.
  dissertation, Massachusetts Institute of Technology, 2001.
\bibitem{laneman}
J.~Laneman, D.~Tse, and G.~W. Wornell, ``Cooperative diversity in wireless
  networks: Efficient protocols and outage behavior,'' \emph{Information
  Theory, IEEE Transactions on}, vol.~50, no.~12, pp. 3062--3080, 2004.

\bibitem{zhao}
Y.~Zhao, R.~Adve, and T.~J. Lim, ``Improving amplify-and-forward relay
  networks: Optimal power allocation versus selection,'' in \emph{Information
  Theory, 2006 IEEE International Symposium on}, 2006, pp. 1234--1238.

\bibitem{borade}
S.~Borade, L.~Zheng, and R.~Gallager, ``Amplify-and-forward in wireless relay
  networks: Rate, diversity, and network size,'' \emph{Information Theory, IEEE
  Transactions on}, vol.~53, no.~10, pp. 3302--3318, 2007.

\bibitem{Leung}
S.~Leung-Yan-Cheong and M.~Hellman, ``{The Gaussian Wire-tap Channel},''
  \emph{IEEE Transactions on Information Theory}, vol.~24, no.~4, pp. 451--456,
  Jul 1978.

\bibitem{liang2008}
Y.~Liang, H.~V. Poor, and S.~Shamai, ``{Secure Communication Over Fading
  Channels},'' \emph{IEEE Transactions on Information Theory}, vol.~54, no.~6,
  pp. 2470--2492, 2008.

\bibitem{gopala}
P.~Gopala, L.~Lai, and H.~El~Gamal, ``{On the Secrecy Capacity of Fading
  Channels},'' \emph{IEEE Transactions on Information Theory}, vol.~54, no.~10,
  pp. 4687 --4698, Oct. 2008.

\bibitem{parada}
P.~Parada and R.~Blahut, ``{Secrecy Capacity of SIMO and Slow Fading
  Channels},'' in \emph{Proceedings of International Symposium on Information
  Theory (ISIT'05)}, Sept. 2005, pp. 2152 --2155.

\bibitem{khisti}
A.~Khisti and G.~W. Wornell, ``Secure transmission with multiple antennas i:
  The misome wiretap channel,'' \emph{Information Theory, IEEE Transactions
  on}, vol.~56, no.~7, pp. 3088--3104, 2010.

\bibitem{shaf}
S.~Shafiee, N.~Liu, and S.~Ulukus, ``{Towards the Secrecy Capacity of the
  Gaussian MIMO Wire-Tap Channel: The 2-2-1 Channel},'' \emph{IEEE Transactions
  on Information Theory}, vol.~55, no.~9, pp. 4033 --4039, Sept. 2009.

\bibitem{li2009}
J.~Li and A.~P. Petropulu, ``{Transmitter Optimization for Achieving Secrecy
  Capacity in Gaussian MIMO Wiretap Channels},'' \emph{CoRR}, vol.
  abs/0909.2622, 2009.

\bibitem{lai07}
L.~Lai and H.~El~Gamal, ``{Cooperative Secrecy: The Relay-Eavesdropper
  Channel},'' in \emph{IEEE International Symposium on Information Theory, 2007
  (ISIT 2007)}, Jun 2007, pp. 931--935.

\bibitem{dong09}
L.~Dong, Z.~Han, A.~Petropulu, and H.~Poor, ``{Improving Wireless Physical
  Layer Security via Cooperating Relays},'' \emph{IEEE Transactions on Signal
  Processing}, vol.~58, no.~3, pp. 1875 --1888, March 2010.

\bibitem{dong}
{Lun Dong and Zhu Han and Petropulu, A.P. and Poor, H.V.},
  ``Amplify-and-forward based cooperation for secure wireless communications,''
  in \emph{Acoustics, Speech and Signal Processing, 2009. ICASSP 2009. IEEE
  International Conference on}, 2009, pp. 2613--2616.

\bibitem{zhang10}
J.~Zhang and M.~C. Gursoy, ``Collaborative relay beamforming for secrecy,'' in
  \emph{Communications (ICC), 2010 IEEE International Conference on}.\hskip 1em
  plus 0.5em minus 0.4em\relax IEEE, 2010, pp. 1--5.

\bibitem{yang2013cooperative}
Y.~Yang, Q.~Li, W.-K. Ma, J.~Ge, and P.~Ching, ``{Cooperative Secure Beamforming for AF Relay Networks With Multiple Eavesdroppers},'' \emph{IEEE
  Signal Processing Letters}, vol.~20, no.~1, pp. 35--38, 2013.
  
\bibitem{jing}
Y.~Jing and H.~Jafarkhani, ``Network beamforming using relays with perfect
  channel information,'' \emph{ IEEE Transactions on Information Theory},
  vol.~55, no.~6, pp. 2499--2517, 2009.

\bibitem{agnihotri11}
S.~Agnihotri, S.~Jaggi, and M.~Chen, ``Amplify-and-forward in wireless relay
  networks,'' in \emph{Information Theory Workshop (ITW), 2011 IEEE}, 2011, pp.
  311--315.

\bibitem{sarma}
S.~Sarma, S.~Shukla, and J.~Kuri, ``{Joint Scheduling \& Jamming for Data
  Secrecy in Wireless Networks},'' in \emph{2013 11th International Symposium
  on Modeling Optimization in Mobile, Ad Hoc Wireless Networks, (WiOpt 13)},
  2013, pp. 248--255.

\bibitem{raey}
J.~Bar-on and K.~Grasse, ``\BIBforeignlanguage{English}{Global optimization of
  a quadratic functional with quadratic equality constraints},''
  \emph{\BIBforeignlanguage{English}{{Journal of Optimization Theory and
  Applications}}}, vol.~82, no.~2, pp. 379--386, 1994.

\bibitem{ANDREA}
A.~Qualizza, P.~Belotti, and F.~Margot, ``\BIBforeignlanguage{English}{{Linear
  Programming Relaxations of Quadratically Constrained Quadratic Programs}},''
  in \emph{\BIBforeignlanguage{English}{Mixed Integer Nonlinear Programming}},
  ser. The IMA Volumes in Mathematics and its Applications, J.~Lee and
  S.~Leyffer, Eds.\hskip 1em plus 0.5em minus 0.4em\relax Springer New York,
  2012, vol. 154, pp. 407--426.
  \bibitem{horn1985}
  R.~A. Horn and C.~R. Johnson, Eds., \emph{Matrix Analysis}.\hskip 1em plus
    0.5em minus 0.4em\relax New York, NY, USA: Cambridge University Press, 1986.
\bibitem{wilde1964optimum}
    D.~J. Wilde, \emph{{Optimum Seeking Methods}}.\hskip 1em plus 0.5em minus
      0.4em\relax Prentice-Hall Englewood Cliffs, NJ, 1964, vol.~14.
% \bibitem{burden1989numerical}
% R.~L. Burden and J.~D. Faires, \emph{Numerical analysis PWS}.\hskip 1em plus
%   0.5em minus 0.4em\relax Kent Publishing Co. Boston, 1989.
\end{thebibliography}
\end{document}